\documentclass{article}
\usepackage{arxiv}

\usepackage{fancyhdr}
\usepackage[utf8]{inputenc}
\usepackage[numbers, sectionbib]{natbib}
\usepackage{amsmath, amssymb, amsfonts, amsthm}
\usepackage{mathtools}
\usepackage{graphicx}
\usepackage{color, xcolor}
\usepackage{float}
\usepackage{hyperref}
\usepackage{comment}
\usepackage{algorithm, algpseudocode}
\usepackage{enumerate}
\usepackage{enumitem}
\usepackage{subfigure}
\usepackage{times}
\usepackage{sidecap}
\usepackage{tcolorbox}
\usepackage{blindtext}
\usepackage[symbol]{footmisc}

\usepackage{boxedminipage}
\usepackage{xspace}

\theoremstyle{plain}
\newtheorem{theorem}{Theorem}[section]
\newtheorem{lemma}[theorem]{Lemma}

\theoremstyle{definition}

\newtheorem{remark}[theorem]{Remark}

\newtheorem*{def*}{Definition}
\newtheorem*{prfthm*}{Proof of Theorem}

\newcommand{\remove}[1]{}

\AtBeginDocument{%
  }

\title{Distributed Butterfly Analysis using Mobile Agents}

\author{
 Prabhat Kumar Chand \\
  Indian Statistical Institute\\
  Kolkata,\\
  India. \\
  \texttt{pchand744@gmail.com} \\
   \And
 Apurba Das \\
  BITS Pilani Hyderabad\\
  Hyderabad,\\
  India.\\
  \texttt{apurba@hyderabad.bits-pilani.ac.in} \\
  \And
 Anisur Rahaman Molla \\
  Indian Statistical Institute\\
  Kolkata,\\
  India. \\
  \texttt{molla@isical.ac.in} \\
}

\begin{document}
\maketitle
\begin{abstract}

\emph{Butterflies}, or 4-cycles in bipartite graphs, are crucial for identifying cohesive structures and dense subgraphs. While agent-based data mining is gaining prominence, its application to bipartite networks remains relatively unexplored. We propose distributed, agent-based algorithms for \emph{Butterfly Counting} in a bipartite graph $G((A,B),E)$. Agents first determine their respective partitions and collaboratively construct a spanning tree, electing a leader within $O(n \log \lambda)$ rounds using only $O(\log \lambda)$ bits per agent. A novel meeting mechanism between adjacent agents improves efficiency and eliminates the need for prior knowledge of the graph, requiring only the highest agent ID $\lambda$ among the $n$ agents. Notably, our techniques naturally extend to general graphs, where leader election and spanning tree construction maintain the same round and memory complexities. Building on these foundations, agents count butterflies per node in $O(\Delta)$ rounds and compute the total butterfly count of $G$ in $O(\Delta+\min\{|A|,|B|\})$ rounds.
\end{abstract}
\sloppy  

\maketitle


\section{Introduction and Related Literature}\label{intro}

 In a bipartite graph, a $2 \times 2$ biclique is called a {\em butterfly}~\cite{AKP-2017-JCN,SP-2018-WSDM}. Butterfly counting has been extensively studied as it serves as a fundamental structure in bipartite networks~\cite{WLQ+-2020-ICDE}, aiding in community detection~\cite{AKP-2017-JCN} and peeling dense subgraphs~\cite{SP-2018-WSDM}. On the other hand, agents play a crucial role in network exploration across various domains, including underwater navigation~\cite{CGZ2021}, military systems~\cite{LSP2018}, social network modeling~\cite{ZSW2018}, and social epidemiology~\cite{ESS2012}. Agents on bipartite graphs similarly have diverse applications, such as drone-target comodule discovery~\cite{CH+-2020-CC}, where drones identify target groups. Similarly, in fraud detection, software agents analyze bipartite networks in social media to identify fake user profiles, reducing maintenance costs.

Butterfly counting has been vastly studied in both sequential and parallel settings. Sanei-Mehri et al.~\cite{SST-2018-KDD} introduced efficient exact and approximate algorithms, while Wang et al.~\cite{WL+-2019-VLDB} improved counting by prioritizing vertices. Zhou et al.~\cite{ZWC-2021-VLDB} extended this to uncertain bipartite graphs. Sar{\i}y{\"u}ce and Pinar~\cite{SP-2018-WSDM} introduced tip and wing decomposition for dense subgraph discovery. Parallel algorithms include Shi and Shun’s~\cite{SS-2022-MGA} multicore approaches, Xu et al.’s~\cite{XZY-2022-VLDB} GPU algorithm, and Weng et al.’s~\cite{WZ+-2022-TPDS} distributed methods for dynamic graphs.  

The mobile agent model used in the paper has been used for computing MIS~\cite{PramanickSPM23,maximal}, small dominating sets~\cite{run_for_cover}, and MST construction~\cite{manish_mst}. Recent work applied agents for triangle counting and truss decomposition~\cite{aamas_ea} and BFS tree construction~\cite{agent_bfs}. In this paper, we study the butterfly counting problem in an anonymous bipartite graph using autonomous agents.

\subsection{Our Contributions}\label{contrib}


Let $G((A,B),E)$ be an $n$ node arbitrary, simple, connected bipartite graph with a maximum degree $\Delta$ where $n$ mobile agents (with highest ID $\lambda$; known to all the agents) are initially placed at each of the $n$ nodes of $G$ in a dispersed initial configuration\footnote{In a \emph{dispersed initial configuration}, the agents start with at-most one agent at each node. For other starting configurations of agents, dispersion is first achieved in $O(n\log^2n)$ rounds using ~\cite{sudodisc}.}. The agents solve the following problems: 

\begin{enumerate}
    \item \textbf{Leader Election, Spanning Tree and Partition: }
     The agents (i) elect a leader and construct a rooted spanning tree (ii) determine their bi-partition sets ($A$ or $B$), and (ii) determine the size of sets $A$ and $B$ in $O(n \log \lambda)$ rounds, using $O(\log \lambda)$ bits of memory per agent.
    \item \textbf{Counting Node-Based and Total Butterflies: }
    The agents calculate (i) the number of butterflies at each node in $O(\Delta)$ rounds, and (ii) the total number of butterflies in $G$ in additional $O(\min\{|A|,|B|\})$ rounds using $O(\Delta\log \lambda)$ bits of memory per agent.
\end{enumerate}

\section{Model and Problem Statements}\label{model}

\subsection{Model}

\begin{enumerate}
    \item \textbf{Graph:} 
    \begin{itemize}
        \item Connected, undirected, and unweighted bipartite graph $G((A,B),E)$ with $n$ nodes, $|E|=m$ edges and bi-partition $(A,B)$.
        \item Nodes have no identifiers or labels nor any memory to retain information.
        \item Each node $v$ has degree $\delta(v)$; maximum degree of $G$ is $\Delta$.
        \item Outgoing edges at node $v$ are labeled with port numbers in $[0, \delta(v)-1]$.
        
    \end{itemize}

    \item \textbf{Mobile Agents:} 
    \begin{itemize}
        \item $n$ agents, each with a unique ID from $[0, n^{O(1)}]$.
        \item Starting Configuration: One agent per node (\textit{dispersed initial configuration}).
        \item Agents have limited memory but can update and retain information.
        \item Multiple agents can be at the same node (\textit{co-located}) or move through the same edge. Agents cannot stay on the edges. 
        \item Agents recognize entry/exit port numbers but have no visibility beyond their current node.
        \item Only co-located agents can communicate and exchange information.
    \end{itemize}

    \item \textbf{Communication Model and Time Cycle (CCM Model):} Agents operate in synchronous rounds using a global common clock. Each agent follows a \textit{Communicate-Compute-Move (CCM)} cycle per round: \textit{Communicate:} Exchanges information with co-located agents. \textit{Compute:} Processes gathered data within memory constraints, and \textit{Move:} Moves to a neighboring node based on computed decisions.

    \item \textbf{Complexity Measures:} (i) \emph{Time complexity:} Number of rounds to complete the algorithm and (ii) \emph{Memory complexity:} Bits required per agent for execution.
\end{enumerate}

\subsection{Problem Statement}

Let $G((A,B),E)$ be an $n$ node arbitrary, anonymous, simple, connected bipartite graph with a maximum degree $\Delta$. Let $n$ mobile agents with distinct IDs in the range $[0,n^{O(1)}]$, be placed at each of the $n$ nodes of $G$ in a \emph{dispersed} initial configuration. Let $\lambda$ be the highest ID of the agents, i.e., $\lambda\in[0,n^{O(1)}]$. The agents solve the following problems. 

\begin{enumerate}
    \item~\textit{Bi-partition Identification and Partition Size Determination}  - The agents determine the bi-partition to which they belong, calculate the size of each bi-partition and construct a spanning tree rooted at the agent with the smallest ID (leader), which is identified in the process.
    \item~\textit{Butterfly Counting} - The agents compute the number of butterflies associated with each node (local count) and the total number of butterflies (global count) in the graph $G$.
\end{enumerate}

\section{Partition Assignment and Spanning Tree Construction}\label{sec:prelims}


 We assume $ n $ agents start in a \emph{dispersed} configuration, each occupying a node in the bipartite graph $ G((A,B),E) $. This section describes how they identify their partition and compute its size. A leader agent starts the algorithm by assigning itself a partition, and its children continue to explore and assign partitions to their descendants while constructing a spanning tree. When the leader is known, the tree facilitates information dissemination; otherwise, when a leader is not known as a priori, the spanning tree also provides termination, as no global graph parameters are initially known. Our leader election and spanning tree construction extend to general graphs and significantly improve upon prior work~\cite{manish_mst} in a \emph{dispersed} setting. See Section~\ref{sec: without a leader} for details. We first present the algorithm for the known leader case, which serves as the foundation for the no-leader version in the next section. 

\subsection{A Leader Agent is Known}\label{sec: leader known}

Each agent maintains the following variables: $ partition $ (a single-bit value, 0 or 1, with the leader $ \tilde{r} $ set to 0), $ parent $ (stores the port number through which the agent was first discovered ), $ child $ (to record the port number through which the agent has moved to an undiscovered adjacent node to identify new agent), $ completion $ (a boolean, initially \texttt{false}, set to \texttt{true} when all children report completion), and $ nextport $ (port number for the next exploration). A variable $sibling$ is employed to reduce memory on the $child$ pointers which is described in Remark~\ref{rem:mem}. 



 The algorithm starts with the leader $\tilde{r}$ setting its partition to 0 and moving through the smallest available port. Upon discovering a new agent $r$, it assigns $r$ a partition of 1. In the second round, both $\tilde{r}$ and $r$ explore their minimum unvisited ports in parallel, continuing the partition assignment and spanning tree construction. In general, when an agent $r_i$ is visited by $r_j$, it inherits a partition opposite to $r_j$, records the discovery port as its parent, and starts exploring its neighbors. If the partition is already assigned, the visit is ignored. Meanwhile, $r_j$ updates its child pointers accordingly. The strategy which optimizes the exploration is that all agents with assigned partitions explore their next potential children in parallel. An agent $r_i$ explores its neighbors sequentially via $r_i.\texttt{nextport}$, updating it in each round. When an agent completes exploration (it sets $\texttt{nextport} = -1$) and has no children, it sends a \texttt{completion} message, node degree value and partition count (or its sum up-to this children) to its parent. These messages propagate up the tree until $\tilde{r}$ receives all completions, signaling the completion of the exploration. Now, $\tilde{r}$ transmit these values down the tree in a similar way. In the first round, it sends these values to its first child. In the second round, $\tilde{r}$ and its first discovered children send their second and first children, respectively. In this way, these values reach all agents in the graph within $O(n)$ rounds. Once every agent is aware of these values, they use a more efficient down-casting approach as detailed in Section~\ref{sec:comm}.

\begin{lemma}\label{lem:leader known}
    The algorithm in Section~\ref{sec: leader known} assigns a partition to each agent in $O(n)$ rounds.
\end{lemma}

\begin{proof} 
An agent receives its partition when it is first visited by a neighboring agent. Let us partition the set of agents into two disjoint subsets: $T_t$, the set of agents that have already been visited and assigned a partition after $t$ phases (each phase consists of two rounds—one for exploration and one for return), and $S_t$, the set of agents yet unvisited after $t$ phases. Without loss of generality, assume that the first $k$ agents were discovered over $l$ distinct phases. i.e., $|T_l| = k$. We aim to show that by the end of the $(k+1)$‑th phase, at least one new agent is visited i.e., $|T_{k+1}|\geq k+1$. For the analysis, we first show how explorations happen till $k^{th}$ phase and for that we divide the $k$ phases into two parts: first, the $l$ phases by which time $k$ agents were already discovered and next the remaining $k-l$ phases. In the first $l$ phases, the number of exploration attempts is at least $1 + 2 + \dots + l = \frac{l(l+1)}{2}$. In the remaining $k - l$ phases, each of the $k$ agents can explore one edge per phase, contributing an additional $(k - l)k$ possible explorations. Therefore, the total number of edge explorations possible within the subgraph induced by the $k$ agents during the first $k$ phases is $\frac{l(l+1)}{2} + (k - l)k$.
Now, the minimum integral value of the expression is $ \frac{k(k+1)}{2} $, which is greater than the number of edges a graph on $ k $ nodes could possibly have. Therefore, after $k$ phases, all internal edges among the $k$ agents must have been explored. This shows that even if the subgraph induced by these $k$ agents is dense, all internal edges are eventually examined. Similarly, in a sparse graph, the agents may have already exhausted all internal edges and begun discovering new agents before reaching the $k$‑phase threshold. In either case, the process ensures that at least one new unvisited agent must be discovered during the $(k+1)$‑th phase, i.e., $|T_{k+1}|\geq k+1$. By applying this argument, we can conclude that all $n$ agents will have received their partition assignments within $n$ phases or $2n$ rounds, completing the proof.
\end{proof}

\subsection{No Leader Agent is Known}\label{sec: without a leader}

When a leader is unknown, we first elect one. A straightforward approach is repeated communication until a leader emerges, but two challenges arise: (i) agents may fail to meet due to simultaneous movement, and (ii) termination is difficult without global parameters like $n$, preventing agents from knowing when all others have been reached. Existing methods~\cite{run_for_cover,aamas_ea,manish_mst} rely on parameters such as $\Delta$ (max degree) and $\lambda$ (max agent ID) for coordination. Since agents meet based on ID bits rather than port order, the algorithm must allow a $\Delta$-round window for interactions, leading to an $O(\Delta)$ round communication window for meeting full neighbors.

We propose a novel approach that enables agents to meet their neighbors one by one in a predefined order. Unlike prior methods, our technique requires no graph parameters except $\lambda$, ensuring agents meet efficiently. It helps achieve leader election and spanning tree construction from a \emph{dispersed} configuration using only $O(\log \lambda)$ memory and completes in $O(n \log \lambda)$ time, significantly improving upon~\cite{manish_mst} (which elects a leader among $ n $ agents in $ O(m) $ time using $ O(n \log n) $ bits from any initial configuration). Before detailing our leader election algorithm, we introduce this meeting protocol that guarantees two agents meet within $O(\log \lambda)$ time.





\subsubsection{A Meeting Protocol for Two Adjacent Agents}

In this section, we propose a deterministic protocol that guarantees a meeting between two adjacent agents within $ O(\log \lambda) $ rounds. This protocol ensures that at least one of the agents moves to the other's location, enabling the agents to meet.


\begin{algorithm}[H]
\caption{Meeting Protocol for Two Adjacent Agents}
\label{alg:meeting_protocol}
\begin{algorithmic}[1]

\Require Two agents, $r_u$ and $r_v$, each with a unique ID, are initially located at adjacent nodes $u$ and $v$, respectively. Node $u$ is connected to $v$ via outgoing port $p_u$ (and enters $v$ via $p_v$). Agent $r_u$ has no knowledge of $r_v$ or $v$, but wants to meet the neighboring agent via port $p_u$.
\Ensure Agent $r_u$ meets the agent located at the neighboring node accessible via port $p_u$ in a span of $4\log\lambda$ rounds. \Comment{A total of $2\log\lambda$ bits and each bit corresponding to $2$ rounds.}

\State Pad $r_u$'s ID with leading zeros to obtain a $\log \lambda$-bit binary string $r_u.b$
\State $r_u.new\_ID \gets \overline{r_u.b} \, || \, r_u.b$ \Comment{Concatenate bitwise complement of $b_u$ with $b_u$; length $2\log \lambda$}
\For{$i = 0$ to $2\log \lambda - 1$}
    \State $current\_bit \gets i{+}1$-th least significant bit of $r_u.new\_ID$
    \If{$current\_bit = 1$}
        \State Move to neighbor $v$ via port $p_u$ and meet with $r_v$ (if present) \Comment{Round $2i + 1$}
        \State Return to original node $u$ \Comment{Round $2i + 2$}
    \Else
        \State Stay at node $u$ during Rounds $2i + 1$ and $2i + 2$
    \EndIf
    \State $i \gets i+1$
\EndFor

\end{algorithmic}
\end{algorithm}

Let agents $r_u$ and $r_v$ initially reside at adjacent nodes $u$ and $v$, with $u$ joining $v$ via outgoing port $p_u$ at $u$. $r_u$ wants to meet its neighbour via $p_u$. Each agent pads its ID to a $\log \lambda$-bit binary string and appends its bitwise complement to the left, forming a $2\log \lambda$-bit $new\_ID$ (Lines 1–2). The protocol runs for $4\log \lambda$ rounds, with two rounds per bit position from LSB to MSB (Lines 3–12). For each bit, if $r_u.new\_ID$ has a $1$, it moves to node $v$ in the first round and returns in the second; otherwise, it stays at $u$. The protocol is local and requires no knowledge of $r_v$’s state. An agent invokes it with any specific port it wants to take. 

Since the IDs are distinct, there exists at least one bit position where one visiting agent has a $1$ and the respective neighbour has a $0$ in their respective $new\_ID$, ensuring that one moves while the other stays—leading to a meeting. The concatenation with the complement is crucial: without it, simultaneous movements might prevent two agents from meeting. Consider an example where agent $r_u$ wants to meet its neighbor $r_v$ via a specific port $p_u$. Suppose $r_u.ID = 0010$ and $r_v.ID = 0110$. The first, second, and fourth bits of their IDs are identical, meaning they will both either stay or move simultaneously — making a meeting impossible. While their third bits differ, $r_u$ has a 0, indicating it remains stationary, and $r_v$ has a 1, meaning it moves. However, $r_v$’s movement is not guaranteed to be toward $r_u$; it could be exploring a different neighbor. Thus, even this differing bit does not ensure a meeting. By concatenating the ID with its complement, $r_u$'s new ID becomes $11010010$ and $r_v$'s becomes $10010110$, the meeting is now guaranteed. Specifically, at the 7th LSB, $r_u$ moves while $r_v$ remains stationary, ensuring a successful meeting.



The agents first elect a leader and construct a single spanning tree rooted at the leader agent. Each agent begins as a separate single node spanning tree, considering itself as the leader. The core idea is to implement the same tree construction methodology from algorithm in Section~\ref{sec: leader known} with a crucial modification; whenever two agents meet, the one with the smaller ID dominates, causing the larger ID agent to give up on its leader status and join the smaller ID’s tree. This process continues until all agents are connected under a single leader with the smallest ID.  

Each agent additionally maintains a variable \texttt{treelabel}, which initially stores its own ID and helps in tracking the spanning tree to which it belongs. At the start, each agent initializes itself as a leader, assigning its own ID to the \texttt{treelabel} variable. Each agent meets its neighbors one by one using Algorithm~\ref{alg:meeting_protocol}. If two agents belong to different spanning trees (\texttt{treelabel} values are different), the one with the smaller \texttt{treelabel} absorbs the other. The larger ID agent updates its \texttt{treelabel} and reconfigures its \texttt{parent} and \texttt{child} pointers to join the new tree. After merging, all agents in the absorbed tree eventually update their variables, ensuring the tree structure is maintained correctly.  The agent with the smallest ID (i.e., the smallest \texttt{treelabel}) continues absorbing other trees until a single spanning tree remains.  Once an agent has received \texttt{completion} messages from all its children, it propagates the message upward. Eventually, when the agent with the least ID receives completion signals from all its children, the algorithm terminates. Once the leader $\tilde{r}$ is elected, the agents can now execute the algorithm in Section~\ref{sec: leader known}.

\begin{lemma}\label{lem:correctness}
   There is a unique agent $r$ satisfying: (i) $r.leader = \texttt{true}$, and (ii) $r.completion = \texttt{true}$.
\end{lemma}

\begin{proof}
    First we prove there is at-least one agent which satisfies these conditions. For this, let us consider the agent with the least ID, say $\tilde{r}$. Since $\tilde{r}$ has the least ID, its tree subsumes all other trees as it ultimately grows out reach all the $n$ nodes. At this point, it start receiving $completion$ signals from the leaf nodes, until finally, its receives these completion signals from its immediate children and terminating the algorithm by setting its own $completion=\texttt{true}$. Also, since it is never subsumed by any other tree, its $leader$ status ultimately remains to be $\texttt{true}$.

    Now, we prove its uniqueness. Assume another agent $r$ sharing the same status of variables $completion$ and $leader$ upon the termination of the algorithm. Now, since the graph is a connected one, there exists a sequence of agents $\tilde{r},r_{i_1},r_{i_2},\dots,r_{i_j},r$. As, the tree belonging to agent $r$ moves towards $\tilde{r}$ via the sequence of agents, there must be an agent $r_{i_k}; 1\leq k\leq j$ where the trees from $\tilde{r}$ and $r$ meet. Now, if $\tilde{r}$ has a lower ID, ultimately, it subsumes $r$'s tree or vice versa. Hence, one of them must give away its leadership status. Hence, there cannot be two agents with $leader=\texttt{true}$ at the end of the execution of this algorithm.
\end{proof}



\begin{lemma}[Time Complexity]\label{lem:time}
    The algorithm in Section~\ref{sec: without a leader} takes $O(n\log\lambda)$ rounds to complete.
\end{lemma}
\begin{proof}
    The time complexity of the algorithm is primarily determined by $\tilde{r}$ (the leader). Since the algorithm operates in the same manner as the algorithm in Section~\ref{sec: leader known} from the perspective of $\tilde{r}$, it has a similar time complexity of $O(n \log \lambda)$ rounds.
\end{proof}

\begin{remark}[Memory Complexity]\label{rem:mem}
    The main memory complexity of the algorithm comes from maintaining the child pointers. Since an agent can have up to $\Delta$ children, an agent would typically require at-least $O(\Delta\log \lambda)$ bits to manage them. However, we can reduce this complexity by introducing a $sibling$ variable. Whenever an agent discovers its first child, it sets its $child$ pointer to the corresponding outgoing port. Now, as it discovers the next child, it does not store it as its children; rather, it moves to its previous (first) child and stores the outgoing port of its new second child into the first $sibling$ variable and so on. This eliminates the need for using $\Delta$ child pointers for a single agent, and the $child$ pointers can be efficiently stored using $O(\log\lambda)$ bits. The variables, $parent, child$ and $nextport$ requires $O(\log\lambda)+O(\log\Delta)=O(\log\lambda)$ bits. While the other variables require $O(1)$ bits. Therefore, we have a memory complexity of $O(\log\lambda)$ bits.

\end{remark} 

\begin{theorem}\label{thm: no leader}
    Given an arbitrary simple connected graph $G$ with $n$ nodes and $n$ agents in a dispersed initial configuration. Then, a leader among the $n$ agents can be elected, and a spanning tree rooted at the leader node can simultaneously be constructed in $O(n\log \lambda)$ rounds using $O(\log n)$ bits of memory per agent. The agents do not require any prior knowledge about any global parameters, except for $\lambda$, which is an upper bound on the IDs of the agents. 
\end{theorem}

The tree construction and leader election methodology used here can also be extended to general graphs. Also, this process can produce spanning trees of large diameter even in a densely connected graph. For example, in a complete graph of $n$ nodes a spanning tree could be a path of length $n-1$ (outer ring) although its diameter is $1$. Therefore, when two farthest agents need to communicate via the spanning tree, it can take $O(n)$ rounds instead of $O(D)$ rounds. However, in bipartite graphs, we see that the diameter of the spanning tree can only be at-most $O(\min\{|A|,|B|\})$, which we note in the following remark.

\begin{remark}\label{rem:diam}
    The diameter of a spanning tree constructed on a bipartite graph can be at-most $2\cdot\min\{|A|,|B|\}$
\end{remark}
\begin{proof}
    Consider the largest path in the spanning tree and assume that it joins $r_i$ and $r_j$. Now, the maximum length of this path can be $2\cdot\min\{|A|,|B|\}$, since, it must alternately connect between the agents of two partitions $A$ and $B$. In addition, since its a path, it cannot visit any node more than once. 
\end{proof}

\section{Butterfly Counting}\label{Sec: Butterfly Counting}

In this section, we present an algorithm for counting butterflies in a bipartite graph $ G((A,B),E) $. We designate the partition containing the leader (least ID agent) $ \tilde{r} $ as $ A $. The algorithm counts butterflies associated with each node in $ A $ through Phases 1 and 2. In the final phase, agents compute the total number of butterflies by aggregating individual counts from each $ a_i \in A $. Before the algorithm starts, we assume that the $ n $ agents are initially dispersed and have executed the algorithm in Section~\ref{sec: without a leader} to identify their partition, determine its size, compute $ \Delta $, and construct the spanning tree of $ G $.

We use the following notation: agents in $ A $ are denoted as $ a_i $, those in $ B $ as $ b_i $; the neighborhood list of $ a_i $ is $ N(a_i) $; the number of common neighbors between $ a_i $ and $ a_j $ is $ N(a_i, a_j) $; the number of butterflies involving $ a_i $ and $ a_j $ is $ B(a_i, a_j) $; and the total number of butterflies containing $ a_i $ is $ B(a_i) $.

\textbf{Phase 1: }Each agent scans its neighborhood and records its neighbors in memory. Before counting begins, agents must establish communication, which requires knowledge of their partition, and $ \Delta $, which is used to synchronize the movements of the agents. Communication follows alternating phases: in odd rounds, agents in $ A $ communicate while those in $ B $ remain stationary; in even rounds, the vice-versa. This ensures that within $ 2\Delta $ rounds, all agents have exchanged information with their neighbors. Each agent requires $ O(\Delta \log \lambda) $ bits of memory for this process, since each agent $a_i$ needs at-most $\Delta$ entries for managing $N(a_i)$ list, with each entry containing an ID of $O(\log \lambda)$ bits. This forms the dominating chunk of the memory for each agent.

\textbf{Phase 2: }In Phase 2, comprising of $2\Delta$ rounds, each agent $ a_i $ counts the butterflies it shares with $ a_j(\neq a_i)$. It revisits its neighbors and, for each $ b_i \in N(a_i) $, determines the number of common neighbors between itself and $ a_j $, using $ b_i $'s recorded neighbor list $ N(b_i) $. After iterating through all $ b_i $, $ a_i $ computes the number of butterflies involving $ (a_i, a_j) $ as 

\[
B(a_i, a_j) = 
\begin{cases} 
\binom{N(a_i, a_j)}{2}, & \text{if } N(a_i, a_j) > 1 \\ 
0, & \text{otherwise}
\end{cases}
\]

  , sums $\sum_jB(a_i,a_j)$, and stores this value in its memory. This gives the total butterfly count involving $a_i$ i.e, $B(a_i)$.


\textbf{Phase 3: } After each $ a_i $ has computed its local butterfly count, it sends this count to the leader $\tilde{r}$ via the constructed spanning tree. Although each child of an agent can simultaneously transmit its value to its parent, requiring at most $ \min\{|A|, |B|\} $ rounds for all values $ B(a_i) $ to reach $\tilde{r}$, the reverse process—disseminating the total count to all agents—requires a more efficient approach, which we detail below. First, $\tilde{r}$ computes the total butterfly count by summing all $ B(a_i) $ values received and dividing it by 2 using $ \min\{|A|, |B|\} $ rounds.

\subsection{Downward Communication Through the Spanning Tree}\label{sec:comm} To efficiently broadcast a value down a tree, we employ the following strategy: Over the next $ 2\min\{|A|, |B|\} $ rounds, each child oscillates between its position and its parent, ensuring information propagation. In the first two rounds, $\tilde{r}$'s immediate children receive the value. In the next two rounds, the grandchildren of $\tilde{r}$ obtain the value, and the process continues. By the end of $ 2\min\{|A|, |B|\} $ rounds, all agents have received this value from $\tilde{r}$.

So, in the next $ 2\min\{|A|, |B|\} $ rounds, all agents have the total butterfly count. In a similar way, the algorithm can compute butterflies for each $ b_i \in B $ by executing Phases 1 and 2 for agents in $ B $.

\begin{remark}\label{rem: for nodes in B}
    The total number of butterflies in a bipartite graph is half the sum of the individual butterfly counts of all nodes in either $ A $ or $ B $~\cite{rectangle, SST-2018-KDD}. 
\end{remark}

\begin{lemma}[Time and Memory Complexity]\label{b_count}
    The agents calculate (i) the number of butterflies based on their own node in $O(\Delta)$ rounds, and (ii) the total number of butterflies in $G$ in $O(\Delta+\min\{|A|,|B|\})$ rounds. 
    \end{lemma}

    \begin{lemma}[Memory Analysis]\label{b_memory}
    The butterfly counting algorithm requires $O(\Delta\log\lambda)$ bits of memory per agent.
    \end{lemma}
    \begin{proof}
        Each agent $a_i$ needs at-most $\Delta$ entries for managing $N(a_i)$ list, with each entry containing an ID of $O(\log \lambda)$ bits. This forms the dominating chunk of the memory for each agent.
    \end{proof}

    \begin{theorem}[Butterfly Counting with Agents]
         Let $G((A,B),E)$ be an $n$ node arbitrary, simple, connected bipartite graph with a maximum degree $\Delta$. Let $n$ mobile agents with distinct IDs in the range $[0,n^{O(1)}]$ with the highest agent ID $\lambda\in[0,n^{O(1)}]$, be placed at each of the $n$ nodes of $G$ in a dispersed initial configuration. Then, the agents can
    \begin{enumerate}
        \item calculate the number of butterflies based on their own node in $O(\Delta)$ rounds, and
        \item calculate the total number of butterflies in $G$ in $O(\Delta+\min\{|A|,|B|\})$ rounds.
    \end{enumerate}
    using $O(\Delta\log \lambda)$ bits of memory per agent.
    \end{theorem}

\section{Future Directions}\label{sec: conclusion}

In this work, we developed algorithms that enable mobile agents to identify their partition in the bipartite graph, determine partition sizes, and compute a spanning tree. Additionally, we designed a leader election algorithm for $n$ agents, which extends to general graphs. Finally, we proposed algorithms for butterfly counting, a fundamental problem in bipartite graph analysis.  

We aim to extend our work by reducing the number of agents required. The goal is to analyze the trade-off between agent count, time complexity, and memory usage. As a start, we look to determine whether the same problem can be efficiently solved using only $\min\{|A|,|B|\}$ agents while minimizing computational overhead. In addition, we look to use these counting techniques to explore \emph{Tip} and \emph{Wing Decomposition} problems, which are useful in constructing a hierarchy of butterfly dense vertex and edge induced sub-graphs in bipartite networks.   

\bibliographystyle{unsrt}  
\bibliography{references}  


\end{document}